\documentclass[preprint,12pt,AMA,Times1COL]{elsarticle}
\usepackage[english]{babel}
\usepackage{amsthm}

\newtheorem{theorem}{Theorem}[section]
\newtheorem{definition}{Definition}

\usepackage{amssymb}
\usepackage{amsmath}

\begin{document}

\begin{frontmatter}

\title{Towards EXPTIME One Way Functions\\
Bloom Filters, Succinct Graphs, Cliques, \& Self Masking
}

\author{Shlomi Dolev}

\affiliation{organization={Department of Computer Science},
            addressline={Ben-Gurion University of the Negev}, 
            city={Beer-Sheva},
            postcode={8410501}, 
            state={},
            country={Israel}}

\begin{abstract}
Consider graphs of $n$ nodes, and use a Bloom filter of length $2\log^3n$ bits. An edge between nodes $i$ and $j$, with $i<j$, turns on a certain bit of the Bloom filter according to a hash function on $i$ and $j$.
Pick a set of $\log n$ nodes and turn on all the bits of the Bloom filter required for these $\log n$ nodes to form a clique.
As a result, the Bloom filter implies the existence of certain other edges, those edges $(x,y)$, with $x<y$, such that all the bits selected by applying the hash functions to $x$ and $y$ happen to have been turned on due to hashing the clique edges into the Bloom filter. 

Constructing the graph consisting of the clique-selected edges and those edges mapped to the turned-on bits of the Bloom filter can be performed in polynomial time in $n$.

Choosing a large enough polylogarithmic in $n$ Bloom filter yields that the graph has only one clique of size $\log n$, the planted clique.

When the hash function is black-boxed, finding that clique is intractable and, therefore, inverting the function that maps $\log n$ nodes to a graph is not (likely to be) possible in polynomial time. 
\end{abstract}


\end{frontmatter}

\section{Introduction}
The assumption that one-way functions exist is the foundation of computational security. A random input is used to construct an instance of a one-way function,
such that it should be hard to reconstruct the random input given the constructed instance. The construction should be done in polynomial time. Therefore, the reconstruction can be done in nondeterministic polynomial time, considering the correct random input in a nondeterministic fashion and computing an identical instance in polynomial time.

One may wish to design one-way functions for which no
{\em deterministic} or {\em randomized} polynomial time algorithm that reveals the input (or an equivalent input) is known. Note that proving that no such deterministic (or randomized) algorithm exists yields $P \neq NP$. 

Succinct representation was introduced in \cite{GW83 (1983)} and further studied in \cite{PY86 (1986)}, proving that graph instances that encode 3-SAT imply NEXPTIME in their succinct representations. See also \cite{DFG13 (2013)} for NEXPTIME hard-in-average succinct-permanent instances.

A new Bloom filter-based succinct representation of graphs allows encoding of an $n$ nodes graph $G$, in polylogarithmic time and space, such that $G$ consists of a $\log(n)$\footnote{Here and in the sequel, we use $\log(n)$ as a synonym for $\log_2(n)$.} nodes clique: a clique among a particular set of nodes. 

Such a succinct representation of a graph (\cite{GW83 (1983),PY86 (1986)})
for which an instance of a ($\log(n)$ size) clique problem exists is beyond $P$ under reasonable assumptions (see \cite{CHKX06 (2006)} for the explicit graph case).
Using Bloom filter encoding allows us to plant
(\cite{FGR+17 (2017)}) a clique in an instance of a one-way function that serves as a {\em heuristic} to requiring more than polynomial time to be reversed.
In fact, it establishes a new potential one-way function candidate for commitment schemes.

In addition, we present techniques to mask the plant clique; blackboxing the arrays and hashes of the Bloom filter. 

 Papadimitriou and Yannakakis \cite{PY86 (1986)} wrote:
``Our results suggest that, as a rule, succinct representations have the effect of {\it precisely} exponentiating the complexity (time {\it or} space) of graph properties.'' 

However, their reductions are for the worst cases, structuring specific, succinct representations, which do not imply the average case complexity. Structuring a succinct representation with a planted solution may reveal the solution. In the sequel, we use a random list of $\log(n)$ vertices that participate in the planted clique, together with (a universal) hash function, resulting in a Bloom filter encoding, to define the entire graph; if the definition exposes the clique vertices, then the pre-image can be immediately found, and reversibility is trivial. The instance we produce does not (explicitly) reveal the clique vertices. Moreover, we use self-masking to mask the planted solution further.

We use the randomization of the clique vertices to encode the universal hash function parameters. Similarly to a subset-sum (in the binary field) presented in \cite{CDM22 (2022)}, we prove that an instance has a high probability of having only one pre-image.

\section{Preliminaries on One Way Functions}

The following definition is taken from \cite{GB08 (2008)}.\\

\noindent
A function $f: \{0,1\}^* \rightarrow \{0,1\}^*$ is \textit{one-way} if:
\begin{enumerate}
    \item There exists a probabilistic polynomial time (PPT) algorithm that, on input $x$, outputs $f(x)$;
    \item For every PPT algorithm $A$ there is a negligible function $\nu_A$ such that for sufficiently large $k$,
   \[
    \Pr[f(z) = y \mid x \xleftarrow{s}
    \{0,1\}^k; \,y \leftarrow f(x); \, z \leftarrow A(1^k, y)] \leq \nu_A(k)
    \]
\end{enumerate}

\vspace{0.4cm}

As proof of the existence of a one-way function implies $P \neq NP$, we only prove univalence with high probability. Namely, the probability of a $z$ value different from $x$ is negligible.  

In the sequel, we describe generations of instances that serve as commitments that we believe are hard to reverse, see e.g., \cite{JP00 (2000)} for using a planted clique in explicit, rather than an implicitly succinct represented graph for a cryptographic commitment. 
We use an input random string as a source for random values when generating an instance; in effect, 
randomly select a $\log(n)$ vertices identities set and possibly hash parameters.

Our instance is generally defined by randomly chosen $\log(n)$ distinct vertices' identities, each of $\log(n)$ bits, that we use to construct a clique; in some constructions, the Bloom filter hash function(s) parameters are also part of the instance. Random strings that yield a permutation of the same set of $\log(n)$ vertices identities are regarded as equivalent, as they yield the same set of vertices; we augment the output with a definition of a permutation $perm$ among the chosen vertices relative to their sorted permutation. Therefore, the reverse requirement is merely to find the set of $\log(n)$ vertices' identities used to output an instance.
 
An instance is univalent if the solution exposes the $\log(n)$ distinct vertices' identities that yield the exposed instance using the generation process. 
Univalence with a certain probability is based on calculating the probability of the existence of more than one solution. 

When the random hash parameters are part of the instance, they are obtained from the (suffix of the) input string. The solution, the set of $\log(n)$ distinct vertices identities, coupled with hash parameters of the instance, serves as input for executing the instance generation procedure and should result in the instance.

\section{Bloom Filter Succinct Graphs with a Planted Clique}

{\bf Universal hash functions.} There are several choices for universal hash functions, which mimic random mapping with a low probability of collisions.
Our analysis assumes that the hashing results are uniformly and independently distributed (approximation), comparable to using random results that are uniformly and independently selected.\\

\noindent
$\bullet$ The first one has been 
presented by Carter and Wegman in \cite{CW79 (1979)}:\\
Let $h_{a,b}$ be a universal hash function mapping an integer in the $1,2,\ldots,n^2$ range. In fact, we consider the input for $h_{a,b}$ expressed by $2 \log(n)$ bits used to describe a graph edge between two graph vertices' identities. 
The calculation of the hash value is by the following formula, where $p_h>2 \log^2 (n)$ is a prime\footnote{A convention for choosing the smallest qualified $p_h$ is used whenever $p_h$ value is not explicitly defined.}, 
$a\neq 0$ and $b$ are smaller than $p_h$, and $m=2\log^2(n)$.
\begin{equation}
h_{a,b}(x)=(((ax+b)~{\bmod {~}}p_h)~{\bmod {~}}m)+1
\end{equation}

We do not restrict ourselves to a specific universal hash function, allowing the use of cryptographic universal hash functions, preferring not to use ones that are based on computation security,  e.g., the universal one-way hash functions \cite{NY89 (1989)}, \cite{CMR98 (1998)}, \cite{M02 (2002)}, that may be employed instead of the universal hash function mentioned above as a building block. 
We next describe an alternative example of a hash function based on the Toeplitz matrix  \cite{K95 (1995)}.\\

\noindent
$\bullet$ Toeplitz matrix-based hash, e.g., \cite{K95 (1995)}:\\
Randomly choose a binary vector $r_1$ of length $2\log(n)$ bits as the first row of a matrix $h$ with the number of rows equal to the logarithmic number of entries of the Bloom filter array.
The vector of the second row is obtained by rotating the bits of $r_1$, and in general, the vector of the $i^{th}$ row is obtained by rotating $r_1$, $i-1$ times.
For a given edge $(u,v)$, $u<v$, multiply the concatenation of the binary representations of $u$ and $v$, $u \circ v$ (we use $\circ$ to represent concatenation) by the Toeplitz matrix associated with $h_{r_1}$ to obtain the hash of an edge $(u,v)$ into $SG$.
Note that in the sequel, we typically use a Toeplitz matrix of $2\log^3(n)$ rows and $2\log(n)$ columns. A repetition pattern appears in the matrix because there are more rows than columns. We may define an {\it Toeplitz matrix enhanced hash} by a matrix using (polylogarithmic) $\log^2(n)$ independently randomly chosen rows; each of these rows defines a square Toeplitz matrix; these matrices are concatenated to form the needed nonsquare matrix, avoiding row repetitions. In the sequel, we refer to the standard Toeplitz matrix, though analogous results can be stated in terms of the enhanced version.\\

\noindent
$\bullet$ Polynomial based $k${\it -wise} independent hash function e.g., \cite{WC (1981)}, \cite{PP08 (2008)}:\\
For random $a_0, \ldots ,a_{k-1} \in \{0,\ldots,p-1\}$, the function
\begin{equation}
h(x) = (\Sigma_{i=0}^{k-1}(a_i x^i\bmod p)\bmod m) +1
\end{equation}
where $p_h>m$ is a prime, is $k$-wise independent. Note that the $h(x)$ computation is polynomial in $k$. Thus, when $k$ is polylogarithmic, polylogarithmic invocations of $h(k)$ remain polylogarithmic in computation time. This fact is useful when hashing clique edges among a polylogarithmic number of nodes.\\

Note that there exists a polylogarithmic-size circuit
(using logical gates, AND, NAND, OR, NOR, XOR, NXOR) representation of the constructed Bloom filter that represents the value one (and the address) of each (and therefore of all) of the entries in the Bloom filter.
Thus, the succinct representation using Bloom filter can be stated in terms of succinct circuits, as assumed in   \cite{GW83 (1983)}, \cite{PY86 (1986)}.\\

\begin{definition} \label{def:graphs}
Let $n$ be the number of nodes in a graph $G(V,E)$.
Vertices identifier $i \in V$ respects, $1 \leq i \leq n$, and an edge identifier
$(i,j) \in E$ respects, $1 \leq i < j \leq n$. 
A random string defines the values of the integers $a$ and $b$ for $h_{a,b}$, each expressed by a polylogarithmic number of bits (e.g., $2\log^3(n)$) that are exposed as the parameters of the universal hash function $h_{a,b}$.
In addition, $\log(n)$ integers, each of $\log(n)$ bits, that define the identities of the vertices of a clique with $\log(n)$ vertices in the graph $G$. 

Given a hash function $h_{a,b}$ and $\log(n)$ numbers,
$v_{i1},v_{i2},\ldots,v_{i\log(n)}$ that represent vertices in a clique of  $\log(n)$ vertices, the graph $G$ is defined as follows:

A binary array $SG$  of $2\log^2(n)$ (later we use other polylogarithmic sizes, e.g., $2\log^3(n)$) bits, where all bits are zero but the bits $SG[h_{a,b}(v_{ij},v_{ik})]$,  which are set to have value one.  
The array $SG$ is a succinct representation of $G(V,E)$, where $V=1,2,\ldots,n$ and $E$ is the set of undirected edges $(u,v)$ ($u<v$) for which $SG[h_{a,b}(u,v)]=1$,
where $h_{a,b}(u,v)= h_{a,b}(u+v\log(n))$.

The output consists of $a$ and $b$ and an array $SG$ of $2\log^2(n)$ bits (that correspond to a Bloom filter), such that each bit with value one in $SG$ at index $l$ corresponds to at least one edge $(i,j)$, $i<j$, between two vertices in the clique, for which $h_{a,b}(i,j)=l$. 
\end{definition}

 Note that $a,b$ and the clique edges, defined by the input string, define the bit values of the $SG$. The $SG$ is an array of $2\log^2(n)$ (or e.g., $2\log^3(n)$) bits, such that a bit with index $l$ has a value one iff there is a clique edge $(i,j)$ for which $h_{a,b}(i,j)=l$ (or $h_{r_1}(i,j)=l$ when the Toeplitz matrix is used). 

Note that after the clique edge insertions,
an edge of the $n$ nodes graph exists with probability $\alpha$, where $\alpha$ is the ratio of number of ones in $SG$ over the total number of entries in $SG$. $\alpha$ can be tuned to imply (w.h.p.) the univalent of the graph to consist of no clique of size $\log(n)$ but the one embedded, as described next. 
 
An important issue to note: The probability space of the graphs $G$ defined by the mapping above is not the well-known $G(n,p)$ (for $p=\alpha$), since the probabilities of edges to exist might not be independent. The independence (approximation) is based on the uniform distribution yielded by the universal hash function used to construct $SG$, see, e.g.,   \cite{CW79 (1979)}, 
\cite{LW05 (2005)},
\cite{PP08 (2008)}, \cite{WC (1981)}.

To increase the probability for univalence, namely, one clique with high probability, $SG$ should contain more zeros than ones.
In the sequel, we suggest replacing the $2 \log^2(n)$ bits in $SG$ to $2k_0 \log^2(n)$ bits.\\

\section{Singularity of the log(\lowercase{n}) Planted Clique} 

In this section, we prove that the uniqueness of the (plant) clique solution is implied by $\alpha \leq 1/(2k_0)$, $k_0=\log(n)$. \\

\noindent
Our one-way function is defined by:  

\noindent
$\bullet$ {\it Instance.}   A binary array $SG$ of $2\log^3(n)$ entries and $h$, where $h$ is defined by $a$, $b$ and $p_h$ in the case of the Carter and Wegman universal hash function $h_{a,b}$ or $r_1$ in the case of $h_{r_1}$ of a Toeplitz hash.
In addition, a permutation definition $perm$ for the $\log(n)$ vertices of the clique.\\

\noindent
{\it -- The output of the OWF is the instance.}\\

\noindent
$\bullet$ {\it Generation.} Generate $SG$ a binary array of $2\log^3(n)$ entries. $SG$ encodes $\log(n)$ randomly chosen vertices and randomly chosen hash function parameters  (according to the input random string) having the value one in all entries that correspond, utilizing $h$, to an edge of the clique; all the rest of the entries are zeros.\\

\noindent
{\it -- The input random string of the OWF.} A random string used to define:\\
(1) $\log(n)$ distinct random vertices identities defined by the input random string according to the permutation index described in \cite{DLH13 (2013)} used to define the first $\log(n)$ vertices in a permutation.
Let $perm$ be the permutation of these $\log(n)$ vertices that implies a sorted version of them\footnote{Outputting $perm$ as part of the instance allows exact reversibility to the input random string, once the vertices' clique is revealed.}.
In the sequel, we use {\it Extract Distinct $\log(n)$ Vertices} for the procedure above that is used to extract the $\log(n)$ vertices and $perm$.
(2) The following suffix of the random string defines the hash function parameters, $\log(2\log^3(n))$ bits for $a$, $b$ in case of $h_{a,b}$ or $2\log(n)$ bits for $r_1$. \\

\noindent
$\bullet$ {\it Solution.}   A set of $\log(n)$ vertices arranged according to $perm$ that can be used to generate $SG$ using the generation process above.

\begin{theorem}
\label{t:npc}
    For %randomly structured $SG$ (with a planted clique of size $\log(n)$), 
    the graphs defined in Definition \ref{def:graphs} there exists $\alpha \leq 1/(2k_0)$, where $k_0$ is at most polylogarithmic in $n$,  and $n \geq n_0$, for which the probability of having an additional non-planted clique of size $\log(n)$ encoded by $SG$, assuming independent uniform probability of edge existence, is less than $2(c2^c)/(2c)^{c-1}$
    where $c=\log(n)$.
\end{theorem}
  
\begin{proof}
The proof assumes independence of probabilities of edge existence, as in the standard definition of a random graph $G(n,p)$. The assumption is based on (the approximate) uniform distribution yielded by the universal hash function used to construct $SG$, see, e.g., \cite{CW79 (1979)}, 
\cite{LW05 (2005)},
\cite{PP08 (2008)},
\cite{WC (1981)}\footnote{The polynomial-based hashing in \cite{WC (1981)} can be used to ensure $k${\it -wise} independence among the hashes of $k$ edges of the logaritmic (in the number of nodes) clique.}. 

\noindent
We compute an upper bound on the probability that the output graph contains one or more non-planted cliques of size $\log(n)$. We call such a clique a {\it bad clique}. For each $k\in [1,\log(n)]$, we compute an upper bound on the probability that the graph contains a bad clique with exactly $k$ vertices that are not in the planted clique. We then sum up all the upper bounds to have an upper bound on the probability that the output graph contains a bad clique. 
\\
 \noindent
  $\bullet$ $k=  1$. The bad clique contains exactly $\log(n) -1$  vertices of the planted one and one vertex that is not in the planted clique.  The probability that node $v$ is not in the planted clique and $\log(n)-1$ specific vertices in the planted clique form such a bad clique is $  \alpha^{\log(n)-1} $. Hence, by the union bound, the probability of such a bad clique is bounded from above by $(n-\log(n))\log(n) \alpha^{\log(n)-1}$.\\
  $\bullet$   $k= 2$.  In this case the probability that two nodes $u,v$ not in the planted clique are connected to each other and also to a given subset of $k-2$ nodes in the planted clique is    $\alpha^{2\log(n)-3}$; Since $2\log(n)-3$ non planted edges must be in the graph as follows: $2\log(n)-4$ connecting $u$ and $v$ to $k-2$ nodes in the planted graph, plus the edge $(u,v)$. Using the union bound over all pairs of nodes not in the clique and all subsets of $\log(n)-2$ nodes in the clique, we get:    
  $$\binom{n-\log(n)}{2}\binom{\log(n)}{2}\left (\alpha^{2\log(n)-3}\right )$$\\
  $\bullet$
  For arbitrary $k\in[2,\log(n)]$, the probability that the graph contains a clique consisting of specific $ k$ vertices not in the planted clique and specific $ \log(n)-k$ vertices in the planted clique is 
  $\alpha^{\left(k(\log(n)-k)+\binom{ k}{2}\right)}$. By the union bound over all possible combinations, we get that this probability is bounded by
  $$ \binom{n-\log(n)}{k}\binom{\log(n)} { k} \left(\alpha^{\left(k(\log(n)-k)+\binom{ k}{2}\right)}\right) $$
  Summing up the probabilities for all $k\in[1,\log(n)]$ yields
 \begin{equation}\label{eq:sumall}
 \sum_{k=1}^{\log(n)}\left[\binom{n-\log(n)}{ k}\binom{\log(n)} {k} \left(\alpha^{\left(k(\log(n)-k)+\binom{ k}{2}\right)}\right)\right]
 \end{equation}
  
When $k_0=\log(n)$, $\alpha \leq 1/(2\log(n))$, the elements in the sum can be upper bounded by: $(n \log(n))^k / (2\log(n))^{k(\log(n)-(k+1)/2)}$.   
$n^{k}$ upper bounds the first binomial, and the second binomial is upper bound by $\log^k (n)$. The power of the fraction is $k(\log(n)-(k+1)/2)$.
 
$n=2^c$, yields that the $k$'th element in the sum is $((c2^c)^k)/(2c)^{k(c-(k+1)/2)}$.
For big enough $c$, e.g., $c=64$, the result for $k=1$ is less than $2^{-371}$, for $k=2$, less than $2^{-735}$, for $k=3$, $2^{-1092}$, and for the last element of the sum, $k=c=64$, the probability is less than $2^{-9512}$, thus, approaching zero with the growth of $k$ in the sum.\\

\noindent
The ratio of an element in the sum over the previous element is:

\begin{equation}
[(c2^c)^{(k+1)}/(2c)^{(k+1)(c-((k+1)+1)/2)}]/[(c2^c)^k)/(2c)^{k(c-(k+1)/2)}]=c2^c/(2c)^{c-(k+2)/2}
\end{equation}

Larger the $k$ smaller the ratio, consider the last elements, where $k=c-1$, then the ratio is $c2^c/((2c)^{c-(c+1)/2}))=c2^c/(2c)^{(c/2)-1}$. 
When $c=64$, the smallest ratio is $2^{-147}$, much smaller than $1/2$,
and is smaller than $1/2$ for every $c>8$. 
When the ratio is always less than or equal to $1/2$, we can bound the sum by doubling the first element of the sum.

Thus, the sum of the first $\log(n)$ elements is less than $2(c2^c)/(2c)^{c-1}$.
Therefore, for $n=2^{64}$, the probability for a univalent solution is greater than $1-2^{-370}$. 
%\qed
\end{proof}

Note that when $ k_0=\log(n)$, the number of bits in $SG$ is $2\log^3(n)$, and, therefore, polylogarithmic.
Also, observe that the number of edges in the implicitly defined graph is 
$\Theta(n^2/\log(n))$, exponentially more than the $\Theta(\log^2n)$ edges of the planted clique. 
Note further that our probabilistic analysis for univalence is for the expected case rather than in the worst case, which is a benefit in the scope of one-way functions. 

The following conjecture is based on the quoted citation from \cite{PY86 (1986)}
``Our results suggest that, as a rule, succinct representations have the effect of {\it precisely} exponentiating the complexity (time {\it or} space) of graph properties.'' 

The function defined by the input and output above is a specific (Bloom Filter-based) succinct representation of the $~\log(n)$ clique. A clique of size three, namely, the graph property of having a triangle, is proven in \cite{GW83 (1983)} to be in {\it NP} for both upper and lower bounds when the graph is represented in a certain succinct form.
We conjecture that our succinct representation of the $\log(n)$ clique is also {\it NP} hard.

As the symmetry among nodes' identities yields indifference in the identities to the solution of an instance, to the clique's existence or finding,
$G$ instances are also hard on average.

We use this conjecture in every one of the next sections, assuming the obtained succinct representation of the planted $\log(n)$ clique instances is {\it NP}-hard on average and does not reveal the clique planted in them. We prove that with high probability, no clique other than the planted clique is encoded in an instance.  

In terms of the equivalence sets, the set of edges of $G$ mapped to the $i$'th entry $SG[i]=0$ is an empty set. Each entry $SG[i]$ in $SG$ with a value of one represents a distinct set of edges of $G$ of approximately $(n^2-n)/(4\log^3(n))$ edges. 

Roughly speaking, we next aim to blare the information associated with the information associated with each bit of value one of $GS[i]$ by using several hash functions for mapping each edge to a binary array, and later even avoid the exposure of the parameters of the hash functions. 
The entropy \cite{S48 (1948)} grows as we use a smaller number of bits in the binary array to represent the same number of random bits. Similarly, the Kolmogorov complexity \cite{K68 (1968)} is higher when the number of bits in the binary array is smaller while the number of random bits encoded stays the same, leading to incompressibility.

In the following sections, we reduce the instance number of bits while preserving the high probability for the univalence of the planted clique and keeping the number of edges in the implicit graph exponentially larger than the instance size. In particular, we preserve (or enlarge) the number of edges to be qualified as edges in the explicit graph.  

 Black-boxing heuristics were suggested in e.g., \cite{R11 (2011)}. Consider an oracle $ORC$ that answers queries on whether an edge $(i,j)$ exists in a graph $G$. Consider further that $ORC$ uses the Bloom Filter succinct representation to answer queries. Thus, it eliminates any information that the exposure of the Bloom Filter binary array and the hash functions may reveal. Then, we obtain the following theorem that assumes (the approximation) of independence among edge existence in the explicit graph.

 \begin{theorem}
 The number of queries needed to reveal a $\log(n)$ planted clique by querying the oracle $ORC$ that uses a polylogarithmic Bloom filter succinct representation for answering queries is exponential in the Bloom filter size.  
 \end{theorem}
 \begin{proof}
A sequence of queries may identify a planted clique either by querying $(\log^2(n)-\log(n))/2$ edges that form a clique or by querying and identifying all edges that are not part of the clique, which form a complementary set of edges to the clique. The complementary set strategy requires $\Theta(n^2-\log^2(n))$ queries, which is exponential in the $ORC$ memory. To identify a clique by querying the clique edges, there must exist a query on one of the clique edges, say $(i,j)$, for which a first positive answer is obtained
the probability of querying on the first clique edge and receiving a positive answer is $\Theta(\log^2(n)/n^2)$. Therefore, the expected number of queries is exponential in the succinct representation employed by $ORC$. 
%\qed
\end{proof}

 The following sections are heuristics for the $ORC$ abstraction and usage scenario.  

\section{Towards Black-boxing the Binary Array}
Bloom filters benefit from using several hash functions to reduce the false-positive probability and use an array length that fits the number of hash functions to yield $\alpha=1/2$. Using several binary arrays and hash functions can yield a more balanced number of ones and zeros in the succinct graph representation. Our concern is establishing univalence while avoiding the discovery of the planted clique by polynomial-time algorithms. Hash function black-boxing heuristics (e.g., \cite{R11 (2011)}) serve as our goal rather than reducing the false-positive probability. Thus, we suggest a heuristic black-boxing of the $SG$ array content to have the highest entropy, using several Bloom filters, $SG_1, SG_2,\ldots, SG_f$, each $SG_i$ employing an independently chosen at random hash function 
while keeping the univalence as described next.

Given that the number of elements to be hashed is $|E_c|=(\log^2(n)-\log(n))/2$, we choose $SG_i$, to employ a binary array of size
$m=|E_c|/ \ln(2)$ entries, which is the optimal size when $SG_i$ uses one hash function. Furthermore, the choice of $m$ implies $\alpha=1/2$. 

To reduce the probability of outputting a positive answer, we output a positive answer iff every $SG_i$ Bloom filter outputs a positive answer. 
To keep the arguments of Theorem \ref{t:npc}, we use $f$ Bloom filters to obtain the probability of $\alpha$ for a positive answer. 
A positive answer from the Bloom filter is obtained only when all $f$ hash functions return one.
Thus, $\alpha=((\log^2(n)-\log(n))/2)/(2\log^3(n))=(1/2)^{f}$,
yielding $f=2+\log(\log(n))$, thus, polylogarithmic size of memory, $(|E_c|/\ln(2))f=(|E_c|/\ln(2))(2+\log(\log(n)))$. The generation process is, therefore, polylogarithmic.\\

\noindent
Here, our one-way function is defined by: \\ 

\noindent
$\bullet$ {\it Instance.}  $f$ binary arrays 
$SG_1,SG_2,\ldots,SG_{f}$, each of 
$(\log^2(n)-\log(n))/\ln(2)$ bits, where $f=2+\log(\log(n))$.  
In addition, $f=2+\log(\log(n))$ descriptions (parameters) for hash functions $h_1,h_2,\ldots, h_{f}$. Where the description of $h$ is defined by $a$, $b$, and $p_h$ in the case of the universal Carter and Wegman definition of the hash function $h_{a,b}$ described above, and $r_1$ in the case of Toeplitz, $h_{r_1}$. In addition, a permutation definition $perm$ for the $\log(n)$ vertices of the clique.\\ \\

\noindent
{\it -- The output of the OWF is the instance.}\\

\noindent
$\bullet$ {\it Generation.} 
Randomly generate $v_1,v_2,\ldots,v_{\log(n)}$ distinct vertices each in the range $1$ to $n$. 
Randomly generate the parameters for $h_1,h_2,\ldots,h_{f}$.
Generate $f$ binary arrays 
$SG_1,SG_2,\ldots,SG_{f}$, each of 
$(\log^2(n)-\log(n))/\ln(2)$ bits, where $f=2+\log(\log(n))$. 
Every of $SG_i$s uses $h_{i}$ to encode the edges of the clique defined by $v_1,v_2,\ldots,v_{\log(n)}$; having the value one in all entries that correspond when utilizing $h_{i}$, to an edge of the clique; all the rest of the entries of $SG_i$ are zeros.\\

\noindent
{\it -- The input random string of the OWF.} A random string used to define:\\ 
(1) {\it Extract Distinct $\log(n)$ Vertices} from the random input string, and compute $perm$, their permutation with relation to their sorted permutation.\\
(2) The following suffix of the random strings defines $f$ hash functions parameters, either two $\log(2\log^3(n))$ bits for $a$, $b$ in case of $h_{a,b}$ for each of the $f$ functions, or $2\log(n)$ bits for $r_1$ for each of the $f$ functions. \\

\noindent
$\bullet$ {\it Solution.}   A set of $\log(n)$ vertices that can be used to generate $SG_1,SG_2,\ldots, SG_f$ 
when coupled with $h_1,h_2,\ldots,h_{f}$ using the generation process above\\

Theorem \ref{t:npc} implies univalence after black-boxing the binary array as the probability for finding the value one (and therefore also zero) in the original one array version is identical to returning the value one as a result of finding one in every $SG_i$, when using $h_{i}$ over the queried edge. 

\section{Towards Black-boxing the Hash Parameters} 
Next, we show that the singularity of the planted clique is preserved even when the universal hash parameters are unraveled; a function of the unraveled planted clique edges defines the parameters. Thus, Black-boxing the hash functions also (see, e.g., \cite{R11 (2011)}).  

The number of possible values for the parameters of a universal hash function, in the case of Carter and Wegman definition, is {\it hf}$_1=(2\log^3(n))^2$,
and the number of random bits used to define {\it hf}$_1$ is $2+6\log(\log(n))$ bits, each edge of the clique is represented by $2\log(n)>2+6\log\log(n)$ bits. Thus,
we can use the random bits related to the planted clique, possibly bits that are a function of three distinct clique edges that form a triangle\footnote{Triangles are {\it NP} to be found in succinct graph representation \cite{GW83 (1983)}, which yields an extra difficulty to be discovered in the (masked) succinctly represented instance. Other choices for defining the hash parameters as a function of the random clique vertices are possible.} bitwise xor them to define the parameters of {\it hf}$_1$.

Similarly, the number of possible values for the parameters of a universal hash function, in the case of Toeplitz, is {\it hf}$_2=2^{2\log(n)}$ (in the original form, a small factor more in our suggested enhanced form).  
The number of random bits used to define {\it hf}$_2$ is $2\log(n)$ bits; thus, we can use the xor of the three first (say, according to their integer values) distinct clique edges, each of which is not yet used for defining the hash functions parameters. \\

\noindent
Here, our one-way function is defined by:  \\

\noindent
$\bullet$ {\it Instance.}  $f$ binary arrays 
$SG_1,SG_2,\ldots,SG_{f}$, each of 
$(\log^2(n)-\log(n))/\ln(2)$ bits, where $f=1+\log(\log(n))/2$.
In addition, a permutation definition $perm$ for the $\log(n)$ vertices of the clique.\\

\noindent
{\it -- The output of the OWF is the instance.}\\

\noindent
$\bullet$ {\it Generation.} 
Randomly generate a sorted vector $v_1,v_2,\ldots,v_{\log(n)}$ of distinct vertices each in the range $1$ to $n$. Produce a sorted array, $E_c$ of $(\log^2(n)-\log(n))/2$ entries, $(v_i,v_j)$, edges of the clique.
Let $tae_1,tae_2,\ldots$ be a vector of triangles of distinct clique edges ordered lexicographically by their indices. The number of such triangles is $|E_c|/3$. For $j=0,1,\ldots$ compute $tae_{3j+1} \oplus tae_{3j+2} \oplus tae_{3j+3}$ to generate the parameters for $h_{j+1}$, generating parameters for $h_1,h_2,\ldots,h_{f}$.
As the number of triangles is $(\log^2(n)-\log(n))/6$, that is greater than $f =1+\log(\log(n))/2$, and at most three triangles suffice to encode the parameters $(a,b,p)$ or $r_1$ of a hash function, the parameters of all hash functions can be defined by the triples.

Generate $f$ binary arrays 
$SG_1,SG_2,\ldots,SG_{f}$, each of 
$(\log^2(n)-\log(n))/\ln(2)$ bits, where $f=1+\log(\log(n))/2$. 
Every of $SG_i$s uses $h_{i}$ to encode the edges of the clique defined by $v_1,v_2,\ldots,v_{\log(n)}$; having the value one in all entries that correspond, utilizing $h_{i}$, to an edge of the clique; all the rest of the entries of $SG_i$ are zeros.\\

\noindent
{\it -- The input random string of the OWF.} A random string used to define $\log(n)$ distinct vertices using the
{\it Extract Distinct $\log(n)$ Vertices} over the random input string.\\

\noindent
$\bullet$ {\it Solution.}   A set of $\log(n)$ vertices, $v_1,v_2,\ldots,v_{\log(n)}$ that can be used to generate the instance $SG_1,SG_2,\ldots, SG_f$ 
when coupled with the $h_1,h_2,\ldots,h_{f}$ that are generated from the clique edges formed among 
$v_1,v_2,\ldots,v_{\log(n)}$, as defined in the generation process.

\begin{theorem}
There are polylogarithmic in $n$ number of black-boxed arrays and polylogarithmic black-boxed hash function parameters, implying with a high probability that the planted clique is the only solution for an instance.
\end{theorem}

\begin{proof}
Theorem \ref{t:npc} implies univalence when the hash parameters used for planting the clique are considered. Considering all possible non planted cliques $\binom{n}{\log(n)}-1<n^{\log(n)}$, each defines (random but correlated with the random clique vertices identities) hash function parameters that yield a probability $(2\log(n))^{(\log^2(n)-log(n))/2}$, to form a clique.
The additional probability due to black-boxing the hash parameters is bounded by $n^{\log(n)}/ (2\log(n))^{(\log^2(n)-log(n))/2}$ which is $2^{(c^2+c)/2}c^{-(c^2-c)/2}$, where $c=\log(n)$.
When $c=64$, the additional probability due to black-boxing the hash function is $2^{-10016}$.
%\qed
\end{proof}

\section{Self Masking for Further Enhancement of Black-boxing}
Consider an instance being only $SG$ where $SG=SG_1 \oplus SG_2\ldots \oplus SG_f$.
When {\it assuming} symmetry due to the approximate uniform independent probability of the resulting xored bits of $SG$, then the number of possible $SG$'s is $2^{(\log^2(n)-\log(n))/\ln(2)}$.
%=n^{(\ln(n)-\ln(2))/\ln^2(2)}$$
The number of possible cliques of $\log(n)$ vertices is $\binom{n}{\log(n)}$.

Thus, the probability that a clique is mapped to a certain $SG$ is: $\binom{n}{\log(n)}/(2^{(\log^2(n)-\log(n))/\ln(2)})$. 
%which is greater than
%$(n-\log(n))^{\log(n)})/(\log(n))^{\log(n)}(2^{(\log^2(n)-\log(n))2})$.\
When $n=2^{64}$ the probability is upper bounded by $2^{-2044}$. \\

\noindent
Here, our one-way function is defined by:  \\

\noindent
$\bullet$ {\it Instance.}  A binary array 
$SG$, of 
$(\log^2(n)-\log(n))/\ln(2)$ bits.
In addition, a permutation definition $perm$ for the $\log(n)$ vertices of the clique.\\

\noindent
{\it -- The output of the OWF is the instance.}\\

\noindent
$\bullet$ {\it Generation.} 
Randomly generate a sorted vector $v_1,v_2,\ldots,v_{\log(n)}$ of distinct vertices each in the range $1$ to $n$. Produce a sorted array, $E_c$ of $(\log^2(n)-\log(n))/2$ entries, $(v_i,v_j)$, edges of the clique.
Let $tae_1,tae_2,\ldots$ be a vector of triangles of distinct clique edges ordered lexicographically by their indices. The number of such triangles is $|E_c|/3$. For $j=0,1,\ldots$ compute $tae_{3j+1} \oplus tae_{3j+2} \oplus tae_{3j+3}$ to generate the parameters for $h_{j+1}$, generating parameters for $h_1,h_2,\ldots,h_{f}$.
As the number of triangles is $(\log^2(n)-\log(n))/6$, which is greater than $f =1+\log(\log(n))/2$, the parameters of all hash functions can be defined by the triples.

Generate $f$ binary arrays 
$SG_1,SG_2,\ldots,SG_{f}$, each of 
$(\log^2(n)-\log(n))/\ln(2)$ bits, where $f=1+\log(\log(n))/2$. 
Every of $SG_i$s uses $h_{i}$ to encode the edges of the clique defined by $v_1,v_2,\ldots,v_{\log(n)}$; having the value one in all entries that correspond, utilizing $h_{i}$, to an edge of the clique; all the rest of the entries of $SG_i$ are zeros.

Generate $SG=G_1 \oplus SG_2 \oplus \ldots SG_{f}$.\\

\noindent
{\it -- The input random string of the OWF.} A random string used to define $\log(n)$ distinct vertices using the
{\it Extract Distinct $\log(n)$ Vertices} over the random input string.\\

\noindent
$\bullet$ {\it Solution.}   A set of $\log(n)$ vertices, $v_1,v_2,\ldots,v_{\log(n)}$ that can be used to generate the instance $SG_1,SG_2,\ldots, SG_f$ 
when coupled with the $h_1,h_2,\ldots,h_{l}$ that are generated from the clique edges formed among 
$v_1,v_2,\ldots,v_{\log(n)}$, as defined in the generation process.
Such that $SG=G_1 \oplus SG_2 \oplus \ldots SG_{f}$.\\

\noindent
{\bf Univalent and Birthday paradox.}
Before analyzing the probability of collision existence being very close to one, by the negligible probability that a value of $SG$ will be associated with a clique, the probability of a clique being associated with a distinct $SG$ is negligibly less than 1.

Still, when considering the birthday paradox, the probability of finding any set (of two or more) cliques that yield the same $SG$ is high as we calculate next, but still, finding the colliding instance can be a challenge. 
A commitment can be based on several instances to be reversed to ensure that these several instances have no collisions, see e.g., \cite{CDM23 (2023)}.

The collision probability is approximately:

$$1 - e^{-((n^2-n)/2)(((n^2-n)/2)-1)/(2\cdot 2^{(\log^2(n)-\log(n))/\ln(2))}}$$

Which is close to one, when, for example, $n=2^{64}$.

\section{Discussion}

One may use exponential preprocessing to construct an exponential lookup table, computing the result for every possible input set of $\log(n)$ vertices; such preprocessing requires memory for storing approximately $2^{(\log^2(n)-\log(n))/2}$ values. Obviously, for every OWF candidate, reversing an instance requires only $O(1)$ operations once such (exponential) preprocessing is done. 

On the other hand, when exponential preprocessing and exponential processing are not feasible at any stage, we speculate that finding the planted cliques is a non-feasible (exponential) challenge to cope with.  

We cannot anticipate all possible algorithms that may be used to reverse an instance; moreover, proving that no such algorithm exists implies $P \neq NP$. 

Lastly, observe that attacks can take a different approach than trying to solve an instance directly; see, e.g., \cite{L89 (1989)} in another context, for the useful great idea of solving different randomly chosen instances to solve a given one; the goal of \cite{L89 (1989)} is proving hardness on average. In our scope, a possible attack may be based on algorithm that uses several, $\log^q(n)$, for some constant $q$, random (sorted) vectors, $cl_i$, each of $\log(n)$ distinct vertices, $cl_1, cl_2, \ldots, cl_{\log^q(n)}$,  yielding polynomial time computation to compute a polynomial number of outputs $inst_1, inst_2, \ldots, inst_{\log^q(n)}$. The attack then may use $(cl_i,inst_i)$ as points in, say, an approximation function $F(x,y)$. $F(x,y)$ attempts, we speculate unsuccessfully, to reflect the result of the construction of the $OWF$. 
If no two $(cl_i,inst_i)$, $(cl_j,inst_i)$, $cl_i\neq cl_j$, despite the birthday paradox analysis, then the definition of $F^{-1}(x,y)$ maybe, but we speculate unlikely be, feasible.\\

\noindent
{\bf Acknowledgments.} 
It is a pleasure to thank 
Shlomo Moran and 
Leonid Levin for the very helpful long discussions, and Moti Yung, Moni Naor, Eylon Yogev, and Oded Margalit for their advice and feedback. In particular, Leonid Levin suggested the use of the Toeplitz-based hash function. Lastly, it is a pleasure to thank Jeffrey Ullman for suggesting the clear abstract wording.

\end{document}